\documentclass[letterpaper, 10pt, conference]{ieeeconf}      

\IEEEoverridecommandlockouts                              
\usepackage[utf8]{inputenc}
\usepackage[english]{babel}
\usepackage[normalem]{ulem}
\usepackage{amsmath}
\usepackage{amssymb}
\usepackage{amsfonts}
\usepackage{centernot}
\usepackage{fouriernc}
\usepackage{algorithm}
\usepackage{algorithmic}
\usepackage{booktabs}
\usepackage{graphicx}
\usepackage{subfig}
\usepackage{float}

\usepackage{hyperref}
\newtheorem{thm}{Theorem}[section]

\newtheorem{df}[thm]{Definition}
\newtheorem{prob}[thm]{Problem}
\newtheorem{rem}[thm]{Remark}
\newtheorem{prop}[thm]{Proposition}

\usepackage{amssymb}
\title{\Large Secure Control under Partial Observability with Temporal Logic Constraints$^*$}
\author{Bhaskar Ramasubramanian$^{1}$, Andrew Clark$^{2}$, Linda Bushnell$^{1}$, and Radha Poovendran$^{1}$
\thanks{$^*$This work was supported by the U.S. Army Research Office, the National Science Foundation, and the Office of Naval Research via Grants W911NF-16-1-0485, CNS-1656981, and N00014-17-S-B001 respectively.}
\thanks{$^{1}$Network Security Lab, Department of Electrical and Computer Engineering, 
University of Washington, Seattle, WA 98195, USA. \newline
        {\tt\small \{bhaskarr, lb2, rp3\}@uw.edu}}%
\thanks{$^{2}$Department of Electrical and Computer Engineering, 
Worcester Polytechnic Institute, Worcester, MA 01609, USA. 
        {\tt\small aclark@wpi.edu}}%
}
\date{}

\begin{document}
\maketitle

\begin{abstract}
This paper studies the synthesis of control policies for an agent that has to satisfy a temporal logic specification in a partially observable environment, in the presence of an adversary. 
The interaction of the agent (defender) with the adversary is modeled as a partially observable stochastic game. 
The search for policies is limited to over the space of finite state controllers, which leads to a tractable approach to determine policies. 
The goal is to generate a defender policy to maximize satisfaction of a given temporal logic specification under any adversary policy. 
We relate the satisfaction of the specification in terms of reaching (a subset of) recurrent states of a Markov chain. 
We then present a procedure to determine a set of defender and adversary finite state controllers of given sizes that will satisfy the temporal logic specification. 
We illustrate our approach with an example. 
\end{abstract}

\section{Introduction}\label{Introduction}

Cyber-physical systems (CPSs) are complex entities in which the working of a physical system is governed by interactions with computing devices and algorithms. 
These systems are ubiquitous \cite{baheti2011cyber}, and vary in scale from power systems to medical devices and robots. 
In applications like self-driving cars and robotics, the systems are expected to work in dynamically changing and potentially dangerous environments with a large degree of autonomy. 
A natural question to ask before solving a problem in this domain is the means by which the environment, goals, and constraints, if any, are specified. 

Markov decision processes (MDPs) \cite{bertsekas2015dynamic, puterman2014markov} have been used to model environments where outcomes depend on both, an inherent randomness in the model (transition probabilities), and an action taken by an agent. 
These models have been extensively used in applications, including in robotics \cite{lahijanian2012temporal} and unmanned aircrafts \cite{temizer2010collision}. 
Formal methods \cite{baier2008principles} are a means to verify the behavior of complex models against a rich set of specifications \cite{lahijanian2015formal}. 
Linear temporal logic (LTL) is a particularly well-understood framework to express properties like safety, liveness, and priority \cite{kress2007s, ding2014optimal}. 
These properties can then be verified using off-the-shelf model solvers \cite{cimatti1999nusmv, kwiatkowska2011prism}. 

The system might be the target of malicious attacks with the aim of preventing it from reaching a goal. 
An attack can be carried out on the physical system, on the computers that control the working of the system, or on communication channels between components of the system. 
Such attacks 
have been reported across multiple application domains like power systems \cite{sullivan2017cyber}, automobiles \cite{shoukry2013non}, water networks \cite{slay2007lessons}, and nuclear reactors \cite{farwell2011stuxnet}. 
Therefore, strategies that are designed to only address modeling and sensing errors and uncertainties may not be optimal in the presence of an intelligent adversary who can manipulate the operation of the system. 

Prior work in verifying the satisfaction of an LTL formula over an MDP or a stochastic game assumes that the states are fully observable. 
In many practical scenarios, this may not be the case. 
For example, a robot might only have an estimate of its current location based on the output of a vision sensor \cite{thrun2005probabilistic}. 
This necessitates the use of a framework that accounts for partial observability. 
For the single-agent case, partially-observable Markov decision processes (POMDPs) can be used to try and solve the problem. 
However, partial observability is a serious limitation in determining an `optimal policy' for an agent. 
This demonstrates the need for techniques to determine approximate solutions. 
Heuristics to approximately solve POMDPs include belief replanning, most likely belief state policy, and entropy weighting \cite{cassandra1996acting}, \cite{kaelbling1998planning}, grid-based methods \cite{brafman1997heuristic}, and point-based methods \cite{kurniawati2008sarsop}. 

A large body of work studies classes of problems that are relevant to this paper (see Sec \ref{RelWork}). 
These can be divided into three broad categories: \emph{i)}: synthesis of strategies for systems represented as an MDP that has to additionally satisfy a TL formula; \emph{ii)}: synthesis of strategies for POMDPs; \emph{iii)}: synthesis of defender and adversary strategies for an MDP under a TL constraint. 
While there has been recent work on the synthesis of controllers for POMDPs under TL specifications, these have largely been restricted to the single-agent case, and do not address the case when there might be an adversary with a competing objective. 

In this paper, we study the problem of determining strategies for an agent that has to satisfy an LTL formula in the presence of an adversary in a partially observable environment. 
The defender and adversary take actions simultaneously, and these jointly influence the transitions of the system. 
Our approach is motivated by the treatment in \cite{sharan2014finite} and \cite{sharan2014formal} which propose the synthesis of parameterized finite state controllers (FSCs) for a POMDP that will maximize the probability of satisfaction of an LTL formula. 
This is an approximate strategy since it refrains from using the entire observation and action histories and uses only the most recent observation in order to determine an action. 
Although this restricts the class of policies that are searched over, FSCs are attractive since they can be used to solve the average reward problem over the infinite horizon \cite{sharan2014formal}. 

\subsection{Contributions}

We extend this setting to include an adversary who is also limited in that it does not exactly observe the state. 
The adversary policy is determined by an FSC, whose goal is opposite to that of the defender. 
The goal for the defender will be to synthesize a policy that will maximize satisfaction of an LTL formula for any adversary policy. 
We show that this is equivalent to maximizing, under any adversary policy,  the probability of reaching a recurrent set of a Markov chain that additionally contains states that need to be reached in order to satisfy the LTL formula. 
The search for policies involve optimizing over both the size of the FSC and its parameters (transition probabilities). 
We present a procedure that will allow for the determining of defender and adversary FSCs of fixed sizes that will satisfy the LTL formula with nonzero probability. 
The search for a defender policy that will maximize the probability of satisfaction of the LTL formula for any adversary policy is then reduced to a search among these FSCs of fixed size. 
If these FSCs are parameterized in an appropriate way, it might lend itself to gradient-based optimization techniques. 

\subsection{Outline}

A quick introduction to LTL and partially observable stochastic games (POSGs) is given in Section \ref{Prelim}. 
We set up our problem in Section \ref{Problem}, where we first define FSCs for the two agents, and show how they can be composed with a POSG to yield a Markov chain. 
Section \ref{Results} presents our main results relating LTL satisfaction on a POSG to reaching recurrent sets of a Markov chain, and a procedure to determine candidate FSCs. 
An illustrative example is presented in Section \ref{Example}. 
Section \ref{RelWork} summarizes related work in POMDPs and TL satisfaction on MDPs, and 
Section \ref{Conclusion} concludes the paper, along with a pointer to future directions of research. 

\section{Preliminaries}\label{Prelim}

In this section, we give a concise introduction to linear temporal logic and partially observable stochastic games. 
We then detail the construction of an entity which will ensure that runs on a POSG will satisfy an LTL formula. 

\subsection{Linear Temporal Logic}\label{LTL}

A \emph{linear temporal logic (LTL) formula} \cite{baier2008principles} is defined over a set of atomic propositions $\mathcal{AP}$, and can be inductively written as: 
$\phi:=\mathtt{T}|\sigma| \neg \phi | \phi \wedge \phi | \mathbf{X} \phi |\phi \mathbf{U} \phi$

Here, $\sigma \in \mathcal{AP}$, and $\mathbf{X}$ and $\mathbf{U}$ are temporal operators denoting the \emph{next} and \emph{until} operations respectively.  

The semantics of LTL are defined over (infinite) words in $2^{\mathcal{AP}}$, and we write $\eta_0 \eta_1\dots:=\eta \models \phi$ when a trace $\eta \in (2^{\mathcal{AP}})^{\omega}$ satisfies an LTL formula $\phi$. 
Further, let $\eta^i = \eta_i\eta_{i+1}\dots$. 
Then, $\eta \models \mathtt{T}$ if and only if (iff) $\eta_0$ is true; 
$\eta \models \sigma$ iff $\sigma \in \eta_0$; 
$\eta \models \neg \phi$ iff $\eta \not \models \phi$; 
$\eta \models \phi_1 \wedge \phi_2$ iff $\eta \models \phi_1$ and $\eta \models \phi_2$; 
$\eta \models \mathbf{X} \phi$ iff $\eta^1 \models \phi$; 
$\eta \models \phi_1 \mathbf{U} \phi_2$ iff $\exists j \geq 0$ such that $\eta^j \models \phi_2$ and for all $k < j, \eta^k \models \phi_1$. 

Further, the logic admits derived formulas of the form: 
\emph{i)}: $\phi_1 \vee \phi_2:=\neg(\neg \phi_1 \wedge \neg \phi_2)$; 
\emph{ii)}: $\phi_1 \Rightarrow \phi_2:= \neg \phi_1 \vee \phi_2$; 
\emph{iii)}: $\mathbf{F}\phi:=\mathtt{T} \mathbf{U} \phi  \text{  (eventually)}$; 
\emph{iv)}: $\mathbf{G} \phi:= \neg \mathbf{F} \neg \phi \text{  (always)}$. 

\begin{df}
A \emph{deterministic Rabin automaton (DRA)} is a quintuple $\mathcal{RA} = (Q, \Sigma, \delta, q_0,F)$ where $Q$ is a nonempty finite set of states, $\Sigma$ is a finite alphabet, $\delta : Q \times \Sigma \rightarrow Q$ is a transition function, $q_0 \in Q$ is the initial state, and $F:=\{(L(i),K(i)\}_{i=1}^M$ is such that $L(i), K(i) \subseteq Q$ for all $i$, and $M$ is a positive integer. 
\end{df}

A \emph{run} of $\mathcal{RA}$ is an infinite sequence of states $q_0q_1\dots$ such that $q_{i} \in \delta (q_{i-1}, \alpha)$ for all $i$ and for some $\alpha \in \Sigma$. 
The run is \emph{accepting} if there exists $(L,K) \in F$ such that the run intersects with $L$ finitely many times, and with $K$ infinitely often. 
An LTL formula $\phi$ over $\mathcal{AP}$ can be represented by a DRA with alphabet $2^{\mathcal{AP}}$  that accepts all and only those runs that satisfy $\phi$. 

\subsection{Partially Observable Stochastic Games}\label{POSG}

\begin{df}\label{SGDefn}
A \emph{stochastic game} \cite{niu2018secure} is a tuple $\mathcal{G}:=(S, U_{def}, U_{adv}, \mathbb{T}, \mathcal{AP}, \mathcal{L})$. 
$S$ is a finite set of states, $s_0 \in S$ is the initial state, $U_{def}$ and $U_{adv}$ are the finite sets of actions of the defender and adversary. 
$\mathbb{T}: S \times U_{def} \times U_{adv}  \times S \rightarrow [0,1]$ encodes $\mathbb{T}(s'|s,u_{def},u_{adv})$, the probability of transition from a state $s$ to a state $s'$ when defender and adversary actions are $u_{def}$ and $u_{adv}$ respectively. 
$\mathcal{AP}$ is a set of atomic propositions, and $\mathcal{L}: S \rightarrow 2^{\mathcal{AP}}$ is a labeling function that maps a state to a subset of atomic propositions that are satisfied in that state. 
\end{df}

A stochastic game can thus be viewed as an extension of Markov decision processes (MDPs) to the case when there is more than one player taking an action. 

When $U_{adv} = \emptyset$ and $|U_{def}| = 1$, $\mathcal{G}$ is a \emph{Markov chain (MC)}. 
For $s, s' \in S$, $s'$ is \emph{accessible} from $s$, written $s \rightarrow s'$, if $\mathbb{T}(s_a|s) \mathbb{T}(s_b|s_a) \dots \mathbb{T}(s_i|s_j) \mathbb{T}(s'|s_i) > 0$ for some (finite subset of) states $s_a,s_b,\dots,s_i,s_j$. 
Equivalently, $s \rightarrow s'$ if there is a positive probability of reaching $s'$ from $s$ in a finite number of steps. 
Two states \emph{communicate}, written $s \leftrightarrow s'$, if $s \rightarrow s'$ and $s' \rightarrow s$. 
\emph{Communicating classes} of states cover the state space of the MC. 
A state is \emph{transient} if there is a nonzero probability of not returning to it when we start from that state, and is \emph{positive recurrent} otherwise. 
If some state in a communicating class is recurrent (transient), then the same holds for all other states in that class. 
Moreover, in a finite state MC, every state is either transient or positive recurrent. 
We refer the reader to \cite{meyn2012markov} for a detailed exposition.  

Partially observable stochastic games (POSGs) extend Definition \ref{SGDefn} to the case when states may not be observable, and each agent could observe the state according to a different observation function. 
This can be viewed as an interpretation of POMDPs to the case when there is more than one player. 

\begin{df}
A \emph{partially observable stochastic game} \\is $\mathcal{SG} := (S, U_{def}, U_{adv}, \mathbb{T}, \mathcal{O}_{def}, \mathcal{O}_{adv}, O_{def}, O_{adv}, \mathcal{AP}, \mathcal{L})$, where $S, U_{def}, U_{adv}, \mathbb{T}, \mathcal{AP}, \mathcal{L}$ are as in Definition \ref{SGDefn}. $\mathcal{O}_{def}, \mathcal{O}_{adv}$ denote the (finite) sets of observations available to the defender and adversary. $O_*:S \times \mathcal{O}_* \rightarrow [0,1]$ encodes $\mathbb{P}(o_*|s)$, where $* \in \{def, adv\}$. 
\end{df}

The functions $O_*$ can be viewed as a means to model imperfect sensing. 
Then, we have $\sum_{o \in \mathcal{O}_*} O_*(o|s) = 1$. 

The information available until time $t$, denoted $\mathfrak{I}_t$, can be inductively defined as: $\mathfrak{I}_0 = S$, $\mathfrak{I}_t = \mathfrak{I}_{t-1} \times U_{def} \times \mathcal{O}_{def} \times U_{adv} \times \mathcal{O}_{adv}$. 
The overall information is $\mathfrak{I} := \cup_t \mathfrak{I}_t$. 

\begin{df}
A \emph{(defender or adversary) policy} for the POSG is a map from the overall information to a probability distribution over the respective action space, i.e. $\mu_*: \mathfrak{I} \times U_* \rightarrow [0,1]$, where $* \in \{def, adv\}$. 
\end{df}

Policies of the form above are called \emph{randomized policies}. 
If $\mu_*:\mathfrak{I}\rightarrow U_*$, it is called a \emph{deterministic policy}. 

In this paper, defender and adversary policies will be determined by probability distributions over transitions in finite state controllers (Sec \ref{FSCs}) that are composed with the POSG. 
This method is chosen because the FSCs when composed with the product-POSG (Sec \ref{ProdPOSG}), will result in a finite state Markov chain. 


\subsection{The Product-POSG}\label{ProdPOSG}

In order to find runs on $\mathcal{SG}$ that would be accepted by a DRA $\mathcal{RA}$ built from an LTL formula $\phi$, we construct a product-POSG. 
This construction is motivated by the product-stochastic game construction in \cite{niu2018secure} and the product-POMDP construction in \cite{sharan2014finite}. 

\begin{df}
Given a POSG $\mathcal{SG}$ and a DRA $\mathcal{RA}$ corresponding to an LTL formula $\phi$, a \emph{product-POSG} is a tuple $\mathcal{SG}^{\phi} = (S^{\phi}, U_{def}, U_{adv}, \mathbb{T}^{\phi}, \mathcal{O}_{def}, \mathcal{O}_{adv}, O_{def}^{\phi}, O_{adv}^{\phi}, F^{\phi}, \mathcal{AP}, \mathcal{L}^{\phi})$. 

Here, $S^{\phi} = S \times Q$, $\mathbb{T}^{\phi}((s',q')|(s,q),u_{def}, u_{adv}) = \mathbb{T}(s'|s,u_{def}, u_{adv})$ iff $\delta(q,\mathcal{L}(s')) = q'$, and $0$ otherwise, 
$O_{*}^{\phi}(o|(s,q)) = O_{*}(o|s)$, 
$F^{\phi} = \{(L^{\phi}(i), K^{\phi}(i))\}_{i=1}^M$ with $L^{\phi}(i), K^{\phi}(i) \subset S^{\phi}$, and $(s,q) \in L^{\phi}(i)$ iff $q \in L(i)$, $(s,q) \in K^{\phi}(i)$ iff $q \in K(i)$, $\mathcal{L}^{\phi}((s,q)) = \mathcal{L}(s)$. 
\end{df}

From the above definition, it is clear that acceptance conditions in the product-POSG depend on the DRA while the transition probabilities of the product-POSG are determined by transition probabilities of the original POSG. 
Therefore, a run on the product-POSG can be used to generate a path on the POSG and a run on the DRA. 
Then, if the run on the DRA is accepting, we say that the product-POSG satisfies the LTL specification $\phi$. 

\section{Problem Setup}\label{Problem}

This section details the construction of finite state controllers (FSCs) for the defender and adversary. 
An FSC for an agent can be interpreted as a policy for that agent. 
When the FSCs are composed with the product-POSG, the resulting entity is a Markov chain. 
We then establish a way to determine satisfaction of an LTL specification on the product-POSG in terms of runs on the composed Markov chain. 
A treatment for the single-agent case when the environment is specified as a POMDP was presented in \cite{sharan2014finite}. 

\begin{figure*}[!t]
\normalsize
\begin{align}
&\mathbb{T}^{\phi, \mathcal{C}_{def}, \mathcal{C}_{adv}}((s',q'),g'_{def},g'_{adv}|(s,q),g_{def},g_{adv}) \label{TransFn} \\ &= \sum_{o \in \mathcal{O}_{def}} \sum_{o' \in \mathcal{O}_{adv}} \sum_{u_{def}} \sum_{u_{adv}} O_{def}(o|s) O_{adv}(o'|s) \mu_{def}(g'_{def}, u_{def}|g_{def}, o) \mu_{adv}(g'_{adv}, u_{adv}|g_{adv}, o') \mathbb{T}^{\phi}((s',q')|(s,q),u_{def}, u_{adv})\nonumber
\end{align}

\begin{align}
O_{def}(o_{def}|s) O_{adv}(o_{adv}|s) \mu_{def}(g'_{def}, u_{def}|g_{def}, o_{def}) \mu_{adv}(g'_{adv}, u_{adv}|g_{adv}, o_{adv}) \mathbb{T}^{\phi}((s',q')|(s,q),u_{def}, u_{adv}) > 0 \label{BigEqn1}
\end{align}

\begin{align}
O_{def}(o_{def}|s) O_{adv}(o_{adv}|s) \mu_{def}(g''_{def}, u_{def}|g_{def}, o_{def}) \mu_{adv}(g''_{adv}, u_{adv}|g_{adv}, o_{adv}) \mathbb{T}^{\phi}((s'',q'')|(s,q),u_{def}, u_{adv}) > 0 \label{BigEqn2}
\end{align}

\hrulefill
\end{figure*}

\subsection{Finite State Controllers}\label{FSCs}

Finite state controllers comprise a finite set of internal states. 
The transitions between any two states is governed by the current observation of the agent. 
A directed cyclic graph of internal states of the FSC will allow for remembering events relevant to taking optimal actions \cite{sharan2014finite}. 
In our setting, we will have two FSCs, one for the defender and another for the adversary. 
We will then limit the search for defender and adversary policies to one over FSCs of fixed cardinality. 

\begin{df}
A \emph{finite state controller for the defender (adversary)}, denoted $\mathcal{C}_{def}$ ($\mathcal{C}_{adv}$) is a tuple $\mathcal{C}_* = (G_*, \mu_*)$, where $G_*$ is a finite set of (internal) states of the controller, $\mu_*: G_* \times \mathcal{O}_* \times G_* \times U_* \rightarrow [0,1]$, written $\mu_*(g'_*, u_*|g_*, o_*)$, is a probability distribution of the next internal state and action, given a current internal state and observation. 
The initial state of $\mathcal{C}_*$ is a probability distribution over $G_*$, and will depend on the initial state of the system. 
Here, $* \in \{def, adv\}$. 
\end{df}


The setup works as follows: Initial states of the FSCs are determined by the initial state of the POSG. 
At each time step, the defender will observe the state of $\mathcal{SG}^{\phi}$ according to $O_{def}$ and will commit to a policy $\mu_{def}(\cdot)$ generated by $\mathcal{C}_{def}$. 
The adversary observes this and the state according to $O_{adv}$ and responds with $\mu_{adv}(\cdot)$ generated by $\mathcal{C}_{adv}$. 
These actions are taken concurrently, and are applied to $\mathcal{SG}^{\phi}$, which transitions to the next state per the distribution $\mathbb{T}^{\phi}(\cdot)$, and the process is repeated. 

\begin{df}
An FSC is \emph{proper} if there is a positive probability of satisfying a given LTL formula in a finite number of steps under this policy on a system represented by a POMDP. 
\end{df}

This is similar to the definition in \cite{hansen2003synthesis}, with the distinction that the terminal state of an FSC in that context will be directly related to Rabin acceptance pairs of a Markov chain formed by composing $\mathcal{C}_{def}$ and $\mathcal{C}_{adv}$ with a product-POSG (Sec \ref{GMC}). 
We will restrict ourselves to proper FSCs for the rest of this paper. 

\subsection{The Global Markov Chain}\label{GMC}

The FSCs $\mathcal{C}_{def}$ and $\mathcal{C}_{adv}$, when composed with $\mathcal{SG}^{\phi}$, will result in a finite-state, fully observable Markov chain. 
To maintain consistency with the literature, we will refer to this as the \emph{global Markov chain (GMC)} \cite{sharan2014finite}. 

\begin{df}
The \emph{global Markov chain} resulting from a product-POSG $\mathcal{SG}^{\phi}$ controlled by FSCs $\mathcal{C}_{def}$ and $\mathcal{C}_{adv}$ is the tuple $\mathcal{M}:=\mathcal{M}^{\phi, \mathcal{C}_{def}, \mathcal{C}_{adv}} = (S^{\phi, \mathcal{C}_{def}, \mathcal{C}_{adv}}, \mathbb{T}^{\phi, \mathcal{C}_{def}, \mathcal{C}_{adv}}, \mathcal{AP}, \mathcal{L}^{\phi, \mathcal{C}_{def}, \mathcal{C}_{adv}})$, where $S^{\phi, \mathcal{C}_{def}, \mathcal{C}_{adv}} = S^{\phi} \times G_{def} \times G_{adv}$, $\mathcal{L}^{\phi, \mathcal{C}_{def}, \mathcal{C}_{adv}}((s,q),g_{def}, g_{adv}) =\mathcal{L}^{\phi}((s,q))$, and $\mathbb{T}^{\phi, \mathcal{C}_{def}, \mathcal{C}_{adv}}$ is given by Equation (\ref{TransFn}).
\end{df}

Similar to $\mathcal{SG}^{\phi}$, the Rabin acceptance condition for $\mathcal{M}$ is: $F^{\phi, \mathcal{C}_{def}, \mathcal{C}_{adv}} = \{(L^{\phi, \mathcal{C}_{def}, \mathcal{C}_{adv}}(i), K^{\phi, \mathcal{C}_{def}, \mathcal{C}_{adv}}(i))\}_{i=1}^M$, with $(s,q,g_{def},g_{adv}) \in L^{\phi, \mathcal{C}_{def}, \mathcal{C}_{adv}}(i)$ iff $(s,q) \in L^{\phi}(i)$ and $(s,q,g_{def},g_{adv}) \in K^{\phi,\mathcal{C}_{def}, \mathcal{C}_{adv}}(i)$ iff $(s,q) \in K^{\phi}(i)$. 

A state  of $\mathcal{M}$ is of the form $\mathfrak{s} = (s,q,g_{def},g_{adv})$. 
A path on $\mathcal{M}$ is a sequence $\pi:=\mathfrak{s}_0 \mathfrak{s}_1 \dots$ such that $\mathbb{T}(\mathfrak{s}_{k+1}|\mathfrak{s}_k) > 0$, where $\mathbb{T}(\cdot)$ here corresponds to the transition probabilities in $\mathcal{M}$. 
A path on $\mathcal{M}$ is accepting if it satisfies the Rabin acceptance condition. 
This corresponds to an execution in $\mathcal{SG}^{\phi}$ controlled by $\mathcal{C}_{def}$ and $\mathcal{C}_{adv}$. 
A probability space over $\mathcal{M}$ is defined in the usual way \cite{baier2008principles}. 

\subsection{System Model} \label{Model}

Consider a discrete-time finite-state system: $x(t+1)=f(x(t),u_{def}(t),u_{adv}(t), \mathtt{v}(t))$, where $\mathtt{v}(t)$ represents a stochastic disturbance. 
This system can be abstracted as an SG with finite state and action spaces using a simulation-based algorithm, similar to that in \cite{niu2018optimal}. 

\subsection{Problem Statement}\label{Stmt}

The goal is to synthesize a defender policy that will maximize the probability of satisfaction of an LTL specification under any adversary policy. 
Clearly, this will depend on the FSCs, $\mathcal{C}_{def}$ and $\mathcal{C}_{adv}$. 
In this paper, we will assume that the size of the adversary FSC is fixed, and known. 
This can be interpreted as one way for the defender to have knowledge of the capabilities of an adversary, which is a reasonable assumption. 
Future work will consider the problem for FSCs of arbitrary sizes. 
Formally, 

\begin{prob}\label{Prob}
Given a partially observable environment and an LTL formula, determine a defender policy specified by a finite state controller that maximizes the probability of satisfying the LTL formula under any adversary policy that is represented as a finite state controller of fixed size $|G_{adv}| = G_A$. 
%
That is, 
\begin{align}
\max_{\mathcal{C}_{def}} \min_{\mathcal{C}_{adv}} \mathbb{P}(\mathcal{SG}^{\phi} \models \phi | \mathcal{C}_{def}, \mathcal{C}_{adv}, |G_{adv}| = G_A)
\end{align}
\end{prob}

Optimizing over $\mathcal{C}_{def}$ and $\mathcal{C}_{adv}$ indicates that the solution will depend on $|G_{def}|$, $\mu_{def}(\cdot)$, and $\mu_{adv}(\cdot)$. 

\section{Results}\label{Results}

\subsection{LTL Satisfaction and Recurrent Sets}\label{LTLRecSet}

Our main result relates the probability of the LTL specification being satisfied by the product-POSG, denoted $\mathcal{SG}^{\phi} \models \phi$, in terms of recurrent sets of the GMC. 
Let $\mathcal{R}:=\mathcal{R}^{\phi,\mathcal{C}_{def},\mathcal{C}_{adv}}$ denote the recurrent states of $\mathcal{M}$ under FSCs $\mathcal{C}_{def}$ and $\mathcal{C}_{adv}$. 
Let $\mathcal{R}^S:= (s,q)$ be the restriction of a recurrent state to a state of $\mathcal{SG}^{\phi}$. 

\begin{prop}\label{Proposn}
$\mathbb{P}(\mathcal{SG}^{\phi} \models \phi) > 0$ if and only if there exists $\mathcal{C}_{def}$ such that for any $\mathcal{C}_{adv}$, there exists a Rabin acceptance pair $(L^{\phi}(i),K^{\phi}(i))$ and an initial state of $\mathcal{M}$, $m_0$, the following conditions hold:
\begin{align}
&K^{\phi}(i) \cap \mathcal{R}^S \neq \emptyset \nonumber \\
m_0&\rightarrow (K^{\phi}(i) \times G_{def} \times G_{adv}) \cap \mathcal{R} \label{QualSat}\\
m_0 &\centernot \rightarrow (L^{\phi}(i) \times G_{def} \times G_{adv}) \cap \mathcal{R} \nonumber 
\end{align}
\end{prop} 

\begin{proof}
If for every $(L^{\phi}(i),K^{\phi}(i))$, at least one of the conditions in Equation (\ref{QualSat}) does not hold, then at least one of the following statements is true: 
\emph{i)}: no state that has to be visited infinitely often is recurrent; 
\emph{ii)}: there is no initial state from which a recurrent state that has to be visited infinitely often is accessible; 
\emph{iii)}: some state that has to be visited only finitely often in steady state is recurrent. 
This means $\mathcal{SG}^{\phi} \not \models \phi$ for all $\mathcal{C}_{def}$. 

Conversely, if all the conditions in Equation (\ref{QualSat}) hold for some $(L^{\phi}(i),K^{\phi}(i))$, then $\mathcal{SG}^{\phi} \models \phi$ by construction. 
\end{proof}

To quantify the satisfaction probability for a defender policy under any adversary policy, assume that the recurrent states of $\mathcal{M}$ are partitioned into recurrence classes $\{R_1,\dots,R_p\}$. 
This partition is maximal, in the sense that two recurrent classes cannot be combined to form a larger recurrent class, and all states within a given recurrent class communicate with each other \cite{sharan2014formal}. 

\begin{df}\label{PhiRecSet}
A recurrent set $R_k$ is \emph{$\phi-$feasible} under FSCs $\mathcal{C}_{def}$ and $\mathcal{C}_{adv}$ if there exists $(L^{\phi}(i), K^{\phi}(i))$ such that $K^{\phi}(i) \cap R^S_k \neq \emptyset$ and $L^{\phi}(i) \cap R^S_k = \emptyset$.  
Let $\phi-RecSets^{\mathcal{C}_{def}, \mathcal{C}_{adv}}$ denote the set of $\phi-$feasible recurrent sets under the respective FSCs. 
\end{df}

Over infinite executions, a path of $\mathcal{M}$ will reach a recurrent set. 
Let $\pi \rightarrow R$ denote the event that such a path will reach a recurrent set. 
Then, Theorem \ref{TheoremRecSet} states that 
Problem \ref{Prob} is equivalent to determining defender FSCs that maximize the probability of reaching $\phi-$feasible recurrent sets of the GMC under any adversary FSC. 

\begin{thm} \label{TheoremRecSet}
\begin{align}
&\max_{\mathcal{C}_{def}} \min_{\mathcal{C}_{adv}} \mathbb{P}(\mathcal{SG}^{\phi} \models \phi | \mathcal{C}_{def}, \mathcal{C}_{adv})\nonumber \\&=\max_{\mathcal{C}_{def}} \min_{\mathcal{C}_{adv}} \sum_{R \in \phi-RecSets^{\mathcal{C}_{def}, \mathcal{C}_{adv}}}\mathbb{P} (\pi \rightarrow R)
\end{align}
\end{thm}

\begin{proof}
Since the recurrence classes are maximal, $\mathbb{P}(\pi \rightarrow (R_1 \cup \dots \cup R_p) ) = \sum_{k=1}^p \mathbb{P}(\pi \rightarrow R_k)$. 
From Definition \ref{PhiRecSet}, a $\phi-$feasible recurrent set will necessarily contain a Rabin acceptance pair. 
Therefore, the probability of $\mathcal{SG}^{\phi}$ satisfying the LTL formula under $\mathcal{C}_{def}$ and $\mathcal{C}_{adv}$ is equivalent to the probability of paths on $\mathcal{M}$ leading to $\phi-$feasible recurrent sets. 
That is, $\mathbb{P}(\mathcal{SG}^{\phi} \models \phi | \mathcal{C}_{def}, \mathcal{C}_{adv}) =  \sum_{R \in \phi-RecSets^{\mathcal{C}_{def}, \mathcal{C}_{adv}}}\mathbb{P} (\pi \rightarrow R)$. 

Then, for some (fixed) $\mathcal{C}_{def}$ (and initial state $\mathfrak{s}$ of $\mathcal{M}$), the minimum probability of satisfying $\phi$ over all adversary FSCs is equal to the minimum probability of reaching a $\phi-$feasible recurrent set. 

The result follows for a maximizing $\mathcal{C}_{def}$. 
\end{proof}

Proposition \ref{Proposn} and Theorem \ref{TheoremRecSet} address a broader class of problems than in Problem \ref{Prob} since they do not assume that the size of the adversary FSC is fixed. 

\subsection{Determining Candidate $\mathcal{C}_{def}$ and $\mathcal{C}_{adv}$}\label{FindFSCs}

If the sizes of $\mathcal{C}_{def}$ and $\mathcal{C}_{adv}$ are fixed, then their design is equivalent to determining the transition probabilities between their internal states. 
We are guided by the treatment in \cite{sharan2014formal}. 
However, our framework differs in that we additionally consider the effect of the presence of an adversary while aiming to satisfy an LTL specification. 

Let the FSC policies $\mu_{def}$ and $\mu_{adv}$ be parameterized by $\Phi_{def}$ and $\Phi_{adv}$ respectively. 
Then, with $* \in \{def,adv\}$, $\mu_*(g'_*,u_*|g_*,o_*):=\mu_*(g'_*,u_*|g_*,o_*, \Phi_*)$. 
Any parameterization of the controller is valid so long as it obeys the laws of probability. 
We use the \emph{softmax parameterization} \cite{sharan2014formal, aberdeen2003policy}, 
since it is convex and its derivative can be easily computed. 
Let $\phi_{g'_*,u_*|g_*,o_*} \in \mathbb{R}$ determine the relative probability of making a transition in the FSC along with taking a corresponding action given an observation. 
Then, the transition probabilities of the FSCs are: 
\begin{align}
\mu_*(g'_*,u_*|g_*,o_*,\Phi)=&\frac{e^{\phi_{g'_*,u_*|g_*,o_*}}}{\sum_{g_* \in G_*} \sum_{u_* \in U_*}e^{\phi_{g'_*,u_*|g_*,o_*}}}
\end{align}

The parameterization considered in Algorithm \ref{algo1} can be viewed as a special case of the softmax parameterization with $\phi_{g'_*,u_*|g_*,o_*} = 0$ for all $g'_*,g_*,o_*,u_*$. 

Define $\mathcal{I}_{*} : G_* \times \mathcal{O}_* \times G_* \times U_* \rightarrow \{0,1\}$, where $\mathcal{I}_*(g', u | g, o) = 1 \Leftrightarrow \mu_*(g',u|g,o) > 0$. 
$\mathcal{I}_*(\cdot)$ then serves to indicate if it is possible for an observation $o$ in a state $g$ of an FSC to transition to $g'$ in the FSC while issuing action $u$. 
We further assume that $\forall (g,o) \in G_* \times \mathcal{O}_*, \exists (g',u) \in G_* \times U_*$ such that  $\mathcal{I}_*(g', u | g, o) = 1$ \cite{sharan2014formal}. 
Let a state in the GMC be denoted $\mathfrak{s}:= (s,q,g_{def},g_{adv})$. 
%
%
%
\begin{algorithm}[!t]
\caption{Generate candidate FSCs $\mathcal{C}_{def}, \mathcal{C}_{adv}$}\label{algo1}
\begin{algorithmic}[1]
\REQUIRE $G_{def}$, $G_{adv}$, $\mathcal{SG}^{\phi}$, $\mathcal{I}_{def}^o$, $\mathcal{I}_{adv}^o$
\ENSURE Set of admissible FSC structures $\mathbb{I}:=(\mathbb{I}_{def}, \mathbb{I}_{adv})$, and transition probabilities, $(\mu_{def}(), \mu_{adv}())$ such that GMC has a $\phi-$feasible recurrent set 
\STATE Induce digraph $\mathcal{G}$ of $\mathcal{M}$ of $\mathbb{SG}^{\phi}$ under $\mathcal{I}_{def}^o$ and $\mathcal{I}_{adv}^o$ as $(\mathfrak{S}, \mathcal{E})$, s.t. $\forall \mathfrak{s}_1, \mathfrak{s}_2 \in \mathfrak{S}: \mathfrak{s}_1 \rightarrow \mathfrak{s}_2 \in \mathcal{E} \Leftrightarrow \mathbb{T}(\mathfrak{s}_2|\mathfrak{s}_1) > 0$. 
\STATE $\mathbb{I}_{def} = \mathbb{I}_{adv} = \emptyset$
\STATE $\mathcal{C} = SCCs(\mathcal{G}) = \{C_1, \dots, C_N\}$
\FOR {$C \in \mathcal{C}$ \AND $(L^{\phi}(i),K^{\phi}(i)) \in F^{\phi}$}
\STATE $Bad_i=\{\mathfrak{s}' \notin C: \exists \mathfrak{s} \in C  \text{ s.t. } \mathfrak{s} \rightarrow \mathfrak{s}'\}$
\STATE $Bad_i= Bad_i \cup (C \cap (L^{\phi}(i) \times G_{def} \times G_{adv}))$
\STATE $Good_i = C \cap (K^{\phi}(i) \times G_{def} \times G_{adv})$
\STATE Set $\mathcal{I}_{*}(g'_{*}, u_*|g_*,o_*) = 1$ for all $g'_*, g_*, u_*, o_*$
\WHILE {$\sum_{g'_*, u_*}\mathcal{I}_{*}(g'_{*}, u_*|g_*,o_*)  > 0 \forall o_*, g_*$ \AND $Bad_i \neq \emptyset$}
\STATE Choose $\mathfrak{s'}=(s',q',g_{def}',g_{adv}') \in Bad_i$, \\$\mathfrak{s}''=(s'',q'',g_{def}'',g_{adv}'') \in Good_i$
\FOR {$\mathfrak{s} = (s,q,g_{def},g_{adv}) \in C \setminus Bad_i$}
\FOR {$u_{def} \in U_{def}$}
\STATE $\mu_*(g'_*,u_*|\Phi_*,g_*,o_*) = \frac{\mathcal{I}_{*}(g'_{*}, u_*|g_*,o_*)}{\sum_{g'_*, u_*}\mathcal{I}_{*}(g'_{*}, u_*|g_*,o_*)}$
\IF {$\exists u_{adv} \in U_{adv}$ Eqn (\ref{BigEqn1}) holds }
\STATE $\mathcal{I}_{def}(g'_{def}, u_{def}|g_{def},o_{def}) \leftarrow 0$ \\$ \forall g'_{def}, g_{def} \in G_{def}$
\ENDIF
\ENDFOR
\FOR {$u_{adv} \in U_{adv}$}
\STATE $\mu_*(g''_*,u_*|\Phi_*,g_*,o_*) = \frac{\mathcal{I}_{*}(g''_{*}, u_*|g_*,o_*)}{\sum_{g''_*, u_*}\mathcal{I}_{*}(g''_{*}, u_*|g_*,o_*)}$
\IF {$\forall u_{def} \in U_{def}$, Eqn (\ref{BigEqn2}) holds}
\STATE $\mathcal{I}_{adv}(g''_{adv}, u_{adv}|g_{adv},o_{adv}) \leftarrow 0$
\ENDIF
\ENDFOR
\ENDFOR
\STATE $Bad_i = Bad_i \setminus \{\mathfrak{s}'\}$
\ENDWHILE
\STATE Compute transition probabilities and construct digraph $\mathcal{G}_{new}$ of GMC of $\mathcal{SG}^{\phi}$ under modified $\mathcal{I}_{def}$ and $\mathcal{I}_{adv}$
\STATE $\mathcal{C}_{new} = SCCs(\mathcal{G}_{new})$
\IF {$\exists \mathfrak{s} \in Good_i \text{ s.t. } \mathfrak{s}$ is \emph{recurrent} in $\mathcal{G}_{new}$}
\STATE $\mathbb{I} = (\mathbb{I}_{def} \cup \mathcal{I}_{def}, \mathbb{I}_{adv} \cup \mathcal{I}_{adv})$
\ENDIF
\ENDFOR
\end{algorithmic}
\end{algorithm}
 
In Algorithm \ref{algo1}, for defender and adversary FSCs with fixed number of states, we determine candidate $\mathcal{C}_{def}$ and $\mathcal{C}_{adv}$ such that the resulting $\mathcal{M}$ will have a $\phi-$feasible recurrent set. 
We start with initial candidate structures $\mathcal{I}_*^o$ and induce the digraph of the resulting GMC (\emph{Line 1}). 
This MC might not contain a $\phi-$feasible recurrent set. 
We first determine the set of communicating classes of the MC, which is equivalent to determining the strongly connected components (SCCs) of the induced digraph (\emph{Line 3}). 
A communicating class of the MC will be recurrent if it is a \emph{sink} SCC of the corresponding digraph. 
The states in $Bad_i$ are those in $C$ that are part of the Rabin accepting pair that has to be visited only finitely many times (and therefore, to be visited with very low probability in steady state) (\emph{Line 6}). 
$Bad_i$ further contains states that can be transitioned to from some state in $C$. 
This is because once the system transitions out of $C$, it will not be able to return to it in order to satisfy the Rabin acceptance condition (\emph{Line 5}) (and hence, $C$ will not be recurrent). 
$Good_i$ contains those states in $C$ that need to be visited infinitely often according to the Rabin acceptance condition (\emph{Line 7}). 

Recall that the agents have access to the actual state only via their individual observations. 
A defender action is forbidden if there exists an adversary action that will allow a transition to a state in $Bad_i$ under observations $o_{def}$ and $o_{adv}$. 
This is achieved by setting corresponding entries in $\mathcal{I}_{def}$ to zero (\emph{Lines 12-17}). 
An adversary action is not useful if for every defender action, the probability of transitioning to a state in $Good_i$ is nonzero under $o_{def}$ and $o_{adv}$. 
This is achieved by setting the corresponding entry in $\mathcal{I}_{adv}$ to zero (\emph{Lines 18-23}). 

The computational complexity of Algorithm \ref{algo1} depends on: \emph{i)}: determining the SCCs. 
This can be done in $\mathbf{O}(|\mathfrak{S}|+|\mathcal{E}|)$ \cite{tarjan1972depth}. 
We have $|\mathfrak{S}| = |S||G_{def}||G_{adv}|$ and $|\mathcal{E}| \leq |\mathfrak{S}|^2$. 
Therefore, the SCCs can be determined in $\mathbf{O}(|S|^2|G_{def}|^2|G_{adv}|^2)$ in the worst case. 
\emph{ii)}: determining the structures in \emph{Lines 9-26}. 
This, in the worst case, is $\mathbf{O}(|\mathfrak{S}|(|\mathcal{O}_{def}+|\mathcal{O}_{adv}|)(|\mathfrak{S}|(|U_{def}|+|U_{adv}|))$. 
Defining $|\mathcal{O}| =|\mathcal{O}_{def}|+|\mathcal{O}_{adv}|$ and $|U|=|U_{def}|+|U_{adv}|$, we have an overall computational complexity of $\mathbf{O}(|S|^2|G_{def}|^2|G_{adv}|^2|\mathcal{O}||U|)$. 

\begin{prop}
Algorithm \ref{algo1} is sound. 
That is, each feasible FSC structure $(\mathcal{I}_{def},\mathcal{I}_{adv})$ in $\mathbb{I}$ will have at least one $\phi-$feasible recurrent set. 
\end{prop}

\begin{proof}
This is by construction. 
The output of the algorithm is a set $\{\mathcal{I}_{def}^i,\mathcal{I}_{adv}^i\}_{i=1}^W$ such that the resulting GMC for each case has a state that is recurrent and has to be visited infinitely often. 
This state, by Definition \ref{PhiRecSet}, belongs to $\phi-RecSet^{\mathcal{C}_{def}^i,\mathcal{C}_{adv}^i}$. 
Moreover, if the algorithm returns a nonempty solution, a solution to Problem \ref{Prob} will exist since we assume that the FSCs are proper. 
\end{proof}

Algorithm \ref{algo1} is \emph{suboptimal} since we only consider the most recent observations of the defender and adversary. 
It is also not complete, since there might be a feasible solution that cannot be determined by the algorithm. 

\begin{rem}
For $\mathcal{C}_{def}$ and $\mathcal{C}_{adv}$ of fixed sizes and structures $\mathcal{I}_{def}$ and $\mathcal{I}_{adv}$, a solution to Problem \ref{Prob} is:
\begin{align}
\max_{\Phi_{def}} \min_{\Phi_{adv}} \mathbb{P}(\mathcal{SG}^{\phi} \models \phi | \mathcal{C}_{def}, \mathcal{C}_{adv})
\end{align}
\end{rem}

This follows from the fact that for fixed FSC sizes and structures, the properties of a set (recurrent or transient) in the GMC will not change. 
What remains then is to choose the transition probabilities appropriately. 
For a softmax parameterization, this computation is presented in \cite{sharan2014formal}, and we omit it for want of space. 

\subsection{Determining Recurrent States to Visit}\label{DetRecState}

Algorithm \ref{algo2} returns a subset of the recurrent states that are consistent with the Rabin acceptance pairs that need to be visited `often' in steady state. 
If there is a reward structure over the states of the GMC that incentivizes visits to $Good_k$, then the expected long-term average reward is equal to the \emph{expected occupation measure} of $Good_k$ \cite{sharan2014finite}. 
Moreover, in the infinite horizon, we can assume that the system has been absorbed in a recurrent set, and the resulting (sub-)Markov chain is irreducible. 
Then, this problem can be solved by viewing it as minimizing an \emph{average cost per stage} problem \cite{bertsekas2015dynamic}. 
\begin{algorithm}[!t]
\caption{Recurrent states to visit in steady state}\label{algo2}
\begin{algorithmic}[1]
\REQUIRE $R_k \in \phi-RecSets^{FSC_{def},FSC_{adv}}, \{L^{\phi}(i),K^{\phi}(i)\}_{i=1}^M$
\ENSURE $Good_k \subseteq R_k$, the set of states that need to be visited `often' in steady-state
\STATE $Good_k = \emptyset$
\FOR {$i=1$ to $M$}
\IF {$((L^{\phi}(i) \times G_{def} \times G_{adv}) \cap R_k = \emptyset)$}
\STATE $Good_k = Good_k \cup (((K^{\phi}(i)\times G_{def} \times G_{adv}) \cap R_k)$
\ENDIF
\ENDFOR
\end{algorithmic}
\end{algorithm}

\section{Example}\label{Example}

%
Assume the state space is given by $S:=\{s_i:i = x+My, x \in \{0,\dots,M-1\}, y \in \{0, \dots, N-1\}\}$. 
This will define an $M \times N$ grid. 
The defender's actions are $U_{def} = \{R, L, U, D\}$ and the adversary's actions are $U_{adv} = \{A, NA\}$, denoting right, left, up, down, attack, and not attack. 
The observations of both agents are $\mathcal{O}_{def} = \mathcal{O}_{adv} = \{correct, wrong\}$, with $O_{def}(correct|s_i) = 0.8 = 1-O_{def}(wrong|s_i)$, and $O_{adv}(correct|s_i) = 0.6 = 1-O_{adv}(wrong|s_i)$. 
Let $\mathcal{AP} = \{unsafe, goal\}$. 
Then, if $\phi = \mathbf{G}\mathbf{F} goal \wedge \mathbf{G} \neg unsafe$, it can be shown that the corresponding DRA will have two states $q_0, q_1$, with $F = (\{\emptyset\}, \{q_1\})$. 
The transition probabilities for $(u_{def}, u_{adv}) = (R, NA)$ and $(R, A)$ are defined below. 
The probabilities for other action pairs can be defined similarly. 
Let $N_{s_i}$ denote the neighbors of $s_i$. 
\begin{align*}
\mathbb{T}(s_j|s_i, R, NA)=
\begin{cases}
0.8 & \text{$j=i+1$, $i+1 \% M \neq 0$}\\ 
\frac{0.2}{|N_{s_i}|} & \text{($s_j \in  \{s_i\} \cup N_{s_i}\setminus \{s_{i+1}\}$), $i+1 \% M \neq 0$}\\
1 & \text{$j =i$ and $i+1 \%M = 0$}\\
\end{cases}
\end{align*}
\begin{align*}
\mathbb{T}(s_j|s_i, R, A)=
\begin{cases}
0.6 & \text{$j=i+1$, $i+1 \% M \neq 0$}\\ 
\frac{0.4}{|N_{s_i}|} &\text{($s_j \in  \{s_i\} \cup N_{s_i}\setminus \{s_{i+1}\}$), $i+1 \% M \neq 0$}\\
1 & \text{$j =i$ and $i+1 \%M = 0$}
\end{cases}
\end{align*}

For this example, let $M = 3, N = 2$. 
Then, $|S| = 6$. 
Let $s_4$ be an unsafe state, and $s_5$ be the goal state. 
This is indicated in Figure \ref{GMCInit}. 
Let $|G_{def}| = 2, |G_{adv}|=1$ for the FSCs. 
Assume that for some initial structures $\mathcal{I}_{def}^0, \mathcal{I}_{adv}^0$ the GMC is given by Figure \ref{GMCInit}. 
The figure also indicates the states in terms of its individual components. 
Assume that the LTL formula $\phi$ is such that the states in green denote those that have to be visited infinitely often in steady state, while those in red must be avoided. 
Therefore $(L^{\phi}, K^{\phi}) = \{(\{\emptyset\}, \{m_1\}), (\{m_3\},\{m_2\})\}$. 
The boxes $C_1, C_2, C_3$ indicate the communicating classes of the graph. 
\begin{figure}[!t]
 \centering
  \includegraphics[width=2.45 in]{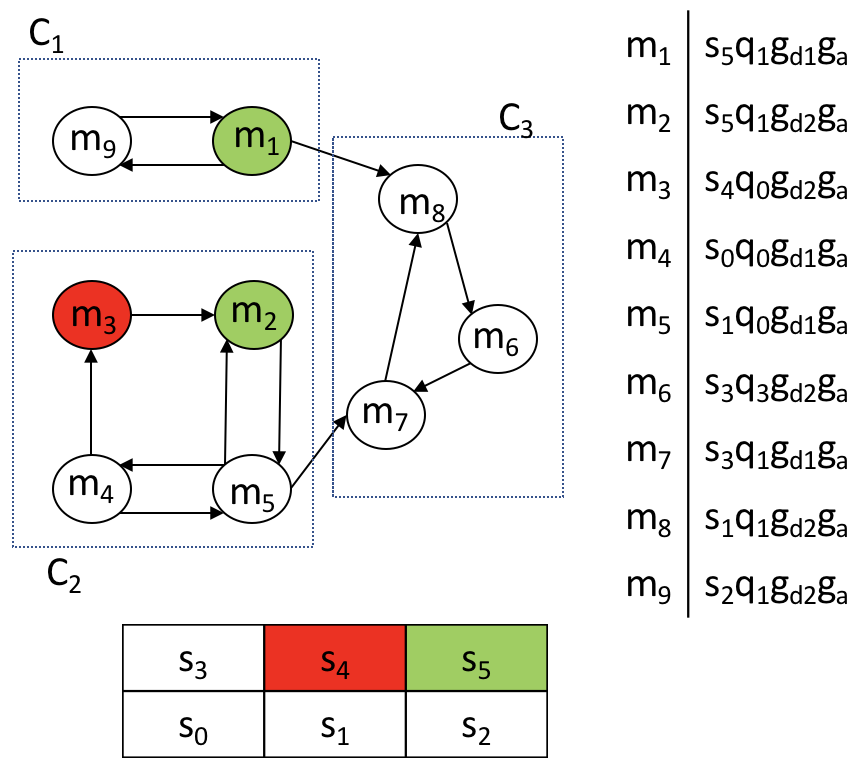} 
\caption{Clockwise, from \emph{top-left}: Global Markov chain (GMC) for initial defender and adversary FSC structures- green states ($m_1 \& m_2$) must be visited infinitely often, and state in red ($m_3$) must be visited finitely often in steady-state; GMC state $m_i \in S \times Q \times G_{def} \times G_{adv}$; State-space for $M=3, N=2$ showing unsafe ($s_4$) and target ($s_5$) states.}\label{GMCInit}
\end{figure}

From Algorithm \ref{algo1}, for $C_1$, $Bad = \{m_8\}, Good = \{m_1\}$. 
For $m_1 \rightarrow m_8$, notice that Equation (\ref{BigEqn1}) is true for all $u_{adv}$ and $u_{def} = \{D, L\}$. 
Therefore, $\mathcal{I}_{def}(g',u_{def}|g,o) \leftarrow 0$ for $o = \{correct, wrong\}$. 
For $m_9 \rightarrow m_1$, since Equation (\ref{BigEqn2}) fails to hold for $R, D \in U_{def}$, $\mathcal{I}_{adv}(\cdot)$ remains unchanged. 
Then, $m_1$ is recurrent in $\mathcal{G}_{new}$. 
For $C_2$, $Bad = \{m_3, m_7\}, Good = \{m_2\}$. 
Like for $C_1$, $\mathcal{I}_{adv}(\cdot)$ remains unchanged, since Equation (\ref{BigEqn2}) does not hold for $D \in U_{def}$. 
Corresponding to $m_5 \rightarrow m_7$, $\mathcal{I}_{def}(g',u_{def}|g,o) \leftarrow 0 \forall u_{def} \in U_{def}\setminus D$. 
A similar conclusion can be reached for $m_4 \rightarrow m_3$. 
Then, $m_2$ will be recurrent in $\mathcal{G}_{new}$. 
For $C_3$, since $Bad = Good = \emptyset$, no structure is added to $\mathbb{I}$. 
Notice that these FSCs satisfy Proposition \ref{Proposn}. 

This example also demonstrates the limitations of Algorithm \ref{algo1}. 
From the $M \times N$ grid, it is clear there will exist a policy that takes the defender from any $s \in S \setminus \{s_4\}$ to $s_5$ with probability $1$. However, for FSCs of small size, the initial state of the defender might result in the Algorithm reporting that no solution was found, even if there exists a feasible solution. 

\section{Related Work}\label{RelWork}

Satisfying TL constraints during motion planning for robots is an active area of research. 
Approaches include hierarchical control \cite{fainekos2009temporal}, ensuring probabilistic satisfaction guarantees \cite{lahijanian2012temporal}, and sensing-based strategies \cite{kress2007s}. 

The authors of \cite{wolff2012robust} propose methods to synthesize a robust control policy that satisfies an LTL formula for a system represented as an MDP whose transitions are not exactly known, but are assumed to lie in a set. 
For MDPs under an LTL specification, a partial ordering on the states is leveraged to solve controller synthesis as a receding horizon problem in \cite{wongpiromsarn2012receding}. 
The synthesis of an optimal control policy that maximizes the probability of an MDP satisfying an LTL formula that additionally minimizes the cost between satisfying instances is studied in \cite{ding2014optimal}. 
This is computed by determining \emph{maximal end components} in an MDP. 
However, this approach will not work in the partially observable setting, where policies will depend on an observation of the state \cite{sadigh2014learning}. 
The synthesis of joint control and sensing strategies for discrete systems with incomplete information and sensing is presented in \cite{fu2016synthesis}. 
The setting of \cite{ding2014optimal} in the additional presence of an adversary with competing objectives has been presented in \cite{niu2018optimal}. 

A policy in a POSG (or POMDP), without loss of generality, depends on the `history' of the system. 
That is, a policy at time $t$ depends on actions and observations at all previous times. 
A memoryless policy on the other hand, only depends on the current state. 
For fully observable stochastic games, it is possible to always find memoryless policies that are optimal. 
However, a policy with memory could perform much better than a memoryless policy for POSGs. 
One way of determining policies for a POSG is to keep track of the entire execution, observation, and action histories, which can be abstracted into determining a sufficient statistic for the POSG execution. 
One example is the \emph{belief state}, which reflects the probability that the agent is in some state, based on receiving observations from the environment. 
Updating the \emph{belief state} at every time step only requires knowledge of the previous belief state and the most recent action and observation. 
Thus, the belief states form the states of an MDP \cite{smallwood1973optimal}, which is more amenable to analysis \cite{bertsekas2015dynamic} than a POMDP.   
However, the belief state is uncountable, and thus will not allow for the development of exact algorithms to determine strategies since these will require nontrivial and potentially infinite memory. 

Synthesis of memoryless strategies for POMDPs in order to satisfy a specification was shown to be NP-hard and in PSPACE in \cite{vlassis2012computational}. 
In \cite{wongpiromsarn2012control}, a discretization of the belief space is carried out \emph{apriori}, resulting in a fully observable MDP. 
However, this approach might not be practical if the state space is large \cite{yu2008near}. 
The complexity of determining a winning strategy to solve the problem of determining the probability of satisfaction of 
parity objectives was shown to be undecidable in \cite{chatterjee2013survey}. 
However, determining finite-memory strategies for the qualitative problem of parity objective satisfaction was shown to be EXPTIME-complete in \cite{chatterjee2014complexity}. 

Dynamic programming (DP) for POSGs has been studied in \cite{hansen2004dynamic}, resulting in an algorithm that generalizes both DP for POMDPs and iterated elimination of dominated strategies for normal form games.
This work, however, considered the finite horizon case, and all agents had to maximize their own expected rewards. 
When agents cooperate to earn rewards, the framework is called a decentralized-POMDP (Dec-POMDP). 
The infinite horizon case for Dec-POMDPs was studied in \cite{bernstein2005bounded}, where the authors proposed a bounded policy iteration algorithm for policies represented as joint FSCs. 
A complete and optimal algorithm for deterministic FSC policies for DecPOMDPs was presented in \cite{szer2005optimal}. 
Optimization techniques for `fixed-size controllers' to solve Dec-POMDPs were investigated in \cite{amato2010optimizing}. 
A survey of recent research in Dec- POMDPs is presented in \cite{oliehoek2016concise}. 

\section{Conclusion}\label{Conclusion}

This paper presented, to the best of our knowledge, the first approach that uses finite state controllers to satisfy an LTL formula in a partially observable environment in the presence of an adversary. 
We showed that the probability of satisfaction of the LTL formula in this setting was equal to the probability of reaching recurrent classes of a Markov chain. 
Further, we presented a procedure to determine defender and adversary controllers of fixed sizes that result in a nonzero satisfaction probability of the LTL formula, and proved its soundness.

%
In ongoing and future work, we plan to investigate the case when the size of the defender FSC can be changed to improve the probability of satisfaction of the LTL formula. 
This has been done for the single agent case in \cite{sharan2014formal}, but extending it to an environment with an intelligent adversary will be challenging and interesting. 
We also plan to study applications of this framework. 

\bibliographystyle{IEEEtran}
\bibliography{PartObsCont.bib}

\begin{thebibliography}{10}
\providecommand{\url}[1]{#1}
\csname url@samestyle\endcsname
\providecommand{\newblock}{\relax}
\providecommand{\bibinfo}[2]{#2}
\providecommand{\BIBentrySTDinterwordspacing}{\spaceskip=0pt\relax}
\providecommand{\BIBentryALTinterwordstretchfactor}{4}
\providecommand{\BIBentryALTinterwordspacing}{\spaceskip=\fontdimen2\font plus
\BIBentryALTinterwordstretchfactor\fontdimen3\font minus
  \fontdimen4\font\relax}
\providecommand{\BIBforeignlanguage}[2]{{%
\expandafter\ifx\csname l@#1\endcsname\relax
\typeout{** WARNING: IEEEtran.bst: No hyphenation pattern has been}%
\typeout{** loaded for the language `#1'. Using the pattern for}%
\typeout{** the default language instead.}%
\else
\language=\csname l@#1\endcsname
\fi
#2}}
\providecommand{\BIBdecl}{\relax}
\BIBdecl

\bibitem{baheti2011cyber}
R.~Baheti and H.~Gill, ``Cyber-physical systems,'' \emph{The Impact of Control
  Technology}, vol.~12, no.~1, pp. 161--166, 2011.

\bibitem{bertsekas2015dynamic}
D.~P. Bertsekas, \emph{Dynamic Programming and Optimal Control 4th Edition,
  Volumes I and II}.\hskip 1em plus 0.5em minus 0.4em\relax Athena Scientific,
  2015.

\bibitem{puterman2014markov}
M.~L. Puterman, \emph{Markov decision processes: {D}iscrete stochastic dynamic
  programming}.\hskip 1em plus 0.5em minus 0.4em\relax John Wiley \& Sons,
  2014.

\bibitem{lahijanian2012temporal}
M.~Lahijanian, S.~B. Andersson, and C.~Belta, ``Temporal logic motion planning
  and control with probabilistic satisfaction guarantees,'' \emph{IEEE
  Transactions on Robotics}, vol.~28, no.~2, pp. 396--409, 2012.

\bibitem{temizer2010collision}
S.~Temizer, M.~Kochenderfer, L.~Kaelbling, T.~Lozano-P{\'e}rez, and J.~Kuchar,
  ``Collision avoidance for unmanned aircraft using {MDP}s,'' in \emph{AIAA
  Guidance, Navigation, and Control Conference}, 2010.

\bibitem{baier2008principles}
C.~Baier and J.-P. Katoen, \emph{Principles of model checking}.\hskip 1em plus
  0.5em minus 0.4em\relax MIT Press, 2008.

\bibitem{lahijanian2015formal}
M.~Lahijanian, S.~B. Andersson, and C.~Belta, ``Formal verification and
  synthesis for discrete-time stochastic systems,'' \emph{IEEE Transactions on
  Automatic Control}, vol.~60, no.~8, pp. 2031--2045, 2015.

\bibitem{kress2007s}
H.~Kress-Gazit, G.~E. Fainekos, and G.~J. Pappas, ``Where's {W}aldo?:
  {S}ensor-based temporal logic motion planning,'' in \emph{International
  Conference on Robotics and Automation}.\hskip 1em plus 0.5em minus
  0.4em\relax IEEE, 2007, pp. 3116--3121.

\bibitem{ding2014optimal}
X.~Ding, S.~L. Smith, C.~Belta, and D.~Rus, ``Optimal control of {MDP}s with
  linear temporal logic constraints,'' \emph{IEEE Transactions on Automatic
  Control}, vol.~59, no.~5, pp. 1244--1257, 2014.

\bibitem{cimatti1999nusmv}
A.~Cimatti, E.~Clarke, F.~Giunchiglia, and M.~Roveri, ``Nusmv: {A} new symbolic
  model verifier,'' in \emph{International Conference on Computer Aided
  Verification}.\hskip 1em plus 0.5em minus 0.4em\relax Springer, 1999, pp.
  495--499.

\bibitem{kwiatkowska2011prism}
M.~Kwiatkowska, G.~Norman, and D.~Parker, ``Prism 4.0: {V}erification of
  probabilistic real-time systems,'' in \emph{International Conference on
  Computer Aided Verification}.\hskip 1em plus 0.5em minus 0.4em\relax
  Springer, 2011, pp. 585--591.

\bibitem{sullivan2017cyber}
J.~E. Sullivan and D.~Kamensky, ``How cyber-attacks in {U}kraine show the
  vulnerability of the {US} power grid,'' \emph{The Electricity Journal},
  vol.~30, no.~3, pp. 30--35, 2017.

\bibitem{shoukry2013non}
Y.~Shoukry, P.~Martin, P.~Tabuada, and M.~Srivastava, ``Non-invasive spoofing
  attacks for anti-lock braking systems,'' in \emph{International Workshop on
  Cryptographic Hardware and Embedded Systems}.\hskip 1em plus 0.5em minus
  0.4em\relax Springer, 2013, pp. 55--72.

\bibitem{slay2007lessons}
J.~Slay and M.~Miller, ``Lessons learned from the {M}aroochy water breach,'' in
  \emph{International Conference on Critical Infrastructure Protection}.\hskip
  1em plus 0.5em minus 0.4em\relax Springer, 2007, pp. 73--82.

\bibitem{farwell2011stuxnet}
J.~P. Farwell and R.~Rohozinski, ``Stuxnet and the future of cyber war,''
  \emph{Survival}, vol.~53, no.~1, pp. 23--40, 2011.

\bibitem{thrun2005probabilistic}
S.~Thrun, W.~Burgard, and D.~Fox, \emph{Probabilistic robotics}.\hskip 1em plus
  0.5em minus 0.4em\relax MIT Press, 2005.

\bibitem{cassandra1996acting}
A.~R. Cassandra, L.~P. Kaelbling, and J.~A. Kurien, ``Acting under uncertainty:
  Discrete bayesian models for mobile-robot navigation,'' in
  \emph{International Conference on Intelligent Robots and Systems},
  vol.~2.\hskip 1em plus 0.5em minus 0.4em\relax IEEE, 1996, pp. 963--972.

\bibitem{kaelbling1998planning}
L.~P. Kaelbling, M.~L. Littman, and A.~R. Cassandra, ``Planning and acting in
  partially observable stochastic domains,'' \emph{Artificial intelligence},
  vol. 101, no. 1-2, pp. 99--134, 1998.

\bibitem{brafman1997heuristic}
R.~I. Brafman, ``A heuristic variable grid solution method for {POMDP}s,'' in
  \emph{AAAI/IAAI}, 1997, pp. 727--733.

\bibitem{kurniawati2008sarsop}
H.~Kurniawati, D.~Hsu, and W.~S. Lee, ``{SARSOP}: Efficient point-based {POMDP}
  planning by approximating optimally reachable belief spaces.'' in
  \emph{Robotics: Science and Systems}, 2008.

\bibitem{sharan2014finite}
R.~Sharan and J.~Burdick, ``Finite state control of {POMDP}s with {LTL}
  specifications,'' in \emph{Proceedings of the American Control Conference},
  2014, pp. 501--508.

\bibitem{sharan2014formal}
R.~Sharan, ``Formal methods for control synthesis in partially observed
  environments: {A}pplication to autonomous robotic manipulation,'' Ph.D.
  dissertation, California Institute of Technology, 2014.

\bibitem{niu2018secure}
L.~Niu and A.~Clark, ``Secure control under {LTL} constraints,'' in
  \emph{Proceedings of the American Control Conference}, 2018.

\bibitem{meyn2012markov}
S.~P. Meyn and R.~L. Tweedie, \emph{Markov chains and stochastic
  stability}.\hskip 1em plus 0.5em minus 0.4em\relax Springer Science \&
  Business Media, 2012.

\bibitem{hansen2003synthesis}
E.~A. Hansen and R.~Zhou, ``Synthesis of hierarchical finite-state controllers
  for {POMDP}s.'' in \emph{ICAPS}, 2003, pp. 113--122.

\bibitem{niu2018optimal}
L.~Niu and A.~Clark, ``Optimal secure control with {LTL} constraints,''
  \emph{Under review}.

\bibitem{aberdeen2003policy}
D.~Aberdeen, ``Policy-gradient algorithms for {POMDPs},'' Ph.D. dissertation,
  The Australian National University, 2003.

\bibitem{tarjan1972depth}
R.~Tarjan, ``Depth-first search and linear graph algorithms,'' \emph{SIAM
  Journal on Computing}, vol.~1, no.~2, pp. 146--160, 1972.

\bibitem{fainekos2009temporal}
G.~E. Fainekos, A.~Girard, H.~Kress-Gazit, and G.~J. Pappas, ``Temporal logic
  motion planning for dynamic robots,'' \emph{Automatica}, vol.~45, no.~2, pp.
  343--352, 2009.

\bibitem{wolff2012robust}
E.~M. Wolff, U.~Topcu, and R.~M. Murray, ``Robust control of uncertain {MDP}s
  with {LTL} specifications,'' in \emph{Proceedings of the Conference on
  Decision and Control}.\hskip 1em plus 0.5em minus 0.4em\relax IEEE, 2012, pp.
  3372--3379.

\bibitem{wongpiromsarn2012receding}
T.~Wongpiromsarn, U.~Topcu, and R.~M. Murray, ``Receding horizon temporal logic
  planning,'' \emph{IEEE Transactions on Automatic Control}, vol.~57, no.~11,
  pp. 2817--2830, 2012.

\bibitem{sadigh2014learning}
D.~Sadigh, E.~S. Kim, S.~Coogan, S.~S. Sastry, and S.~A. Seshia, ``A learning
  based approach to control synthesis of {MDP}s for {LTL} specifications,'' in
  \emph{Proceedings of the Conference on Decision and Control}.\hskip 1em plus
  0.5em minus 0.4em\relax IEEE, 2014, pp. 1091--1096.

\bibitem{fu2016synthesis}
J.~Fu and U.~Topcu, ``Synthesis of joint control and active sensing strategies
  under temporal logic constraints,'' \emph{IEEE Transactions on Automatic
  Control}, vol.~61, no.~11, pp. 3464--3476, 2016.

\bibitem{smallwood1973optimal}
R.~D. Smallwood and E.~J. Sondik, ``The optimal control of {POMDP}s over a
  finite horizon,'' \emph{Operations Research}, vol.~21, no.~5, pp. 1071--1088,
  1973.

\bibitem{vlassis2012computational}
N.~Vlassis, M.~L. Littman, and D.~Barber, ``On the computational complexity of
  stochastic controller optimization in {POMDP}s,'' \emph{ACM Transactions on
  Computation Theory}, vol.~4, no.~4, 2012.

\bibitem{wongpiromsarn2012control}
T.~Wongpiromsarn and E.~Frazzoli, ``Control of probabilistic systems under
  dynamic, partially known environments with temporal logic specifications,''
  in \emph{Proceedings of the Conference on Decision and Control (CDC)}.\hskip
  1em plus 0.5em minus 0.4em\relax IEEE, pp. 7644--7651.

\bibitem{yu2008near}
H.~Yu and D.~P. Bertsekas, ``On near optimality of the set of finite-state
  controllers for average cost {POMDP},'' \emph{Mathematics of Operations
  Research}, vol.~33, no.~1, pp. 1--11, 2008.

\bibitem{chatterjee2013survey}
K.~Chatterjee, L.~Doyen, and T.~A. Henzinger, ``A survey of partial-observation
  stochastic parity games,'' \emph{Formal Methods in System Design}, vol.~43,
  no.~2, pp. 268--284, 2013.

\bibitem{chatterjee2014complexity}
K.~Chatterjee, L.~Doyen, S.~Nain, and M.~Y. Vardi, ``The complexity of
  partial-observation stochastic parity games with finite-memory strategies,''
  in \emph{International Conference on Foundations of Software Science and
  Computation Structures}.\hskip 1em plus 0.5em minus 0.4em\relax Springer,
  2014, pp. 242--257.

\bibitem{hansen2004dynamic}
E.~A. Hansen, D.~S. Bernstein, and S.~Zilberstein, ``Dynamic programming for
  partially observable stochastic games,'' in \emph{AAAI}, vol.~4, 2004, pp.
  709--715.

\bibitem{bernstein2005bounded}
D.~S. Bernstein, E.~A. Hansen, and S.~Zilberstein, ``Bounded policy iteration
  for decentralized {POMDP}s,'' in \emph{International Joint Conference on
  Artificial Intelligence}, 2005, pp. 52--57.

\bibitem{szer2005optimal}
D.~Szer and F.~Charpillet, ``An optimal best-first search algorithm for solving
  infinite horizon {DEC-POMDP}s,'' in \emph{European Conference on Machine
  Learning}.\hskip 1em plus 0.5em minus 0.4em\relax Springer, 2005, pp.
  389--399.

\bibitem{amato2010optimizing}
C.~Amato, D.~S. Bernstein, and S.~Zilberstein, ``Optimizing fixed-size
  stochastic controllers for {POMDP}s and decentralized {POMDP}s,''
  \emph{Autonomous Agents and Multi-Agent Systems}, vol.~21, no.~3, pp.
  293--320, 2010.

\bibitem{oliehoek2016concise}
F.~A. Oliehoek and C.~Amato, \emph{A concise introduction to decentralized
  {POMDP}s}.\hskip 1em plus 0.5em minus 0.4em\relax Springer, 2016.

\end{thebibliography}
\end{document}